\documentclass{article} % For LaTeX2e
\usepackage{nips10submit_e,times}
\usepackage{amsmath,amsfonts,amssymb,amsthm,%%%natbib,
  color,dsfont,graphicx}
\usepackage[ruled,vlined]{algorithm2e}
\newcommand{\X}{\mathbf X}

\newcommand{\B}{\mathbf B}
\newcommand{\one}{\mathds{1}}
\newcommand{\I}{\mathcal I}
\newcommand{\II}{\mathbf I}

\newcommand{\0}{\mathbf 0}
\newtheorem{claim}{Claim}

\newtheorem{thm}{Theorem}
\newtheorem{prob}{Problem}

\newcommand{\V}{\mathbf V}
\newcommand{\W}{\mathbf W}
\newcommand{\U}{\mathbf U}
\newcommand{\D}{\mathbf D}
\newcommand{\bSigma}{\boldsymbol\Sigma}
\newcommand{\s}{\mathbf s}
\newcommand{\bb}{\mathbf b}
\newcommand{\A}{\mathbf A}
\newcommand{\st}{\mathrm{~s.t.~}}

\newcommand{\argmin}{\mathrm{argmin}}
\usepackage{enumerate}

\newcommand{\glpca}{\textsc{gl-SPCA}}
\newcommand{\glreg}{\textsc{gl-Reg}}
\newcommand{\wh}{\widehat}

\renewcommand{\r}{\mathbb{R}}
\newcommand{\n}{\mathcal{N}}
\newcommand{\Err}{\textsc{err}}
\newcommand{\Vcur}{\V_\textsc{cur}}
\newcommand{\Lcol}{\mathcal L_{\mathrm{col}}}
\usepackage{paralist}
\usepackage{psfrag}
\usepackage{multirow, booktabs}

\title{CUR from a Sparse Optimization Viewpoint}

\author{
Jacob~Bien\thanks{Jacob Bien and Ya Xu contributed equally.} \\
Department of Statistics\\
Stanford University\\
Stanford, CA 94305 \\
\small{
\texttt{jbien@stanford.edu} } \\
\And
Ya~Xu$^*$\\
Department of Statistics\\
Stanford University\\
Stanford, CA 94305 \\
\small{
\texttt{yax.stanford@gmail.com} } \\
\And
Michael~W.~Mahoney \\
Department of Mathematics\\
Stanford University\\
Stanford, CA 94305 \\
\small{
\texttt{mmahoney@cs.stanford.edu} } \\
}

\nipsfinalcopy % Uncomment for camera-ready version

\begin{document}

\maketitle

\vspace{-4mm}
\begin{abstract}
The CUR decomposition provides an approximation of a matrix $\X$ that has low
reconstruction error and that is sparse in the sense that the
resulting approximation lies in the span of only a  few  columns of $\X$.
In this
regard, it appears to be similar to many sparse PCA methods.
However, CUR takes a randomized algorithmic
approach, whereas most sparse PCA methods are framed as convex
optimization problems.  In this paper, we try to
understand CUR from a sparse optimization viewpoint.
We show that CUR is implicitly optimizing a sparse
regression objective and, furthermore, cannot be
directly cast
as a sparse PCA method. We also observe that the sparsity attained by CUR
possesses an interesting structure, which leads us to
formulate a sparse PCA method
that
achieves a CUR-like sparsity.
\end{abstract}

\vspace{-4mm}
\section{Introduction}
\label{sec:introduction}
\vspace{-2mm}

CUR decompositions are a recently-popular class of randomized algorithms 
that approximate a data matrix $\X\in\r^{n\times p}$ by using only a small 
number of actual columns of $\X$~\cite{CUR_PNAS,DMM08_CURtheory_JRNL}.
CUR decompositions are often described as SVD-like low-rank decompositions 
that have the additional advantage of being easily interpretable to domain 
scientists.
The motivation to produce a more interpretable low-rank decomposition is 
also shared by sparse PCA (SPCA) methods, which are optimization-based 
procedures that have been of interest recently in statistics and machine 
learning.  

Although CUR and SPCA methods start with similar motivations, they proceed 
very differently.  
For example, most CUR methods have been randomized, and they take a purely 
algorithmic approach.  
By contrast, most SPCA methods start with a combinatorial optimization
problem, and they then solve a relaxation of this problem.  
Thus far, it has not been clear to researchers how the CUR and SPCA 
approaches are related.  
It is the purpose of this paper to understand CUR decompositions from a 
sparse optimization viewpoint, thereby elucidating the connection between 
CUR decompositions and the SPCA class of sparse optimization methods.

To do so, we begin by putting forth a combinatorial optimization problem 
(see \eqref{eq:allsubsets} below)
which CUR is implicitly approximately optimizing. 
This formulation will highlight two interesting features of CUR: 
first, CUR attains a distinctive pattern of sparsity, which has practical 
implications from the SPCA viewpoint; and
second, CUR is implicitly optimizing a regression-type objective. 
These two observations then lead to the three main contributions of this 
paper: 
(a) first, we formulate a non-randomized optimization-based version of CUR 
(see Problem 1: \glreg~ in Section~\ref{sec:cur-optim-fram}) that is based 
on a convex relaxation of the CUR combinatorial optimization problem; 
(b) second, we show that, in contrast to the original PCA-based motivation 
for CUR, CUR's implicit objective cannot be directly expressed in terms of 
a PCA-type objective (see Theorem~\ref{thm:cur-not-pca} in
Section~\ref{sec:connections-pca}); and
(c) third, we propose an SPCA approach (see Problem 2: \glpca~ in 
Section~\ref{sec:group-lasso-pca}) that achieves the sparsity structure of 
CUR within the PCA framework.  
We also provide a brief empirical evaluation of our two proposed objectives.
While our proposed \glreg~ and \glpca~ methods are promising in and of 
themselves, our purpose in this paper is not to explore them as 
alternatives to CUR;
instead, our goal is to use them to help clarify the connection between 
CUR and SPCA methods.

We conclude this introduction with some remarks on notation.  Given a
matrix $\A$, we use $\A_{(i)}$ to denote its $i$th row (as a row-vector) and
$\A^{(i)}$ its $i$th column.  Similarly, given a set of indices $\I$,
$\A_{\I}$ and $\A^{\I}$ denote the submatrices of $\A$ containing only
these $\I$ rows and columns, respectively.  Finally, we let $\Lcol(\A)$ denote the column space of $\A$.

\vspace{-3mm}
\section{Background}
\label{sec:background}
\vspace{-1mm}

In this section, we provide a brief background on CUR and
SPCA methods, with a particular emphasis on topics to which we will return
in subsequent sections.
Before doing so, recall that, given an input matrix $\X$, 
Principal Component Analysis (PCA) seeks the $k$-dimensional hyperplane 
with the lowest reconstruction error.
That is, it computes a $p\times k$ orthogonal matrix $\W$ that minimizes
\begin{align}\label{eq:PCA}
\Err(\W) = ||\X-\X\W\W^T||_F.
\end{align}
Writing the SVD of $\X$ as $\U\bSigma\V^T$, the minimizer of \eqref{eq:PCA} 
is given by $\V_k$, the first $k$ columns of $\V$.
In the data analysis setting, each column of $\V$ provides a particular
linear combination of the columns of $\X$.  
These linear combinations are often thought of as latent factors.  
In many applications, interpreting such factors is made much easier if they 
are comprised of only a small number of actual columns of $\X$, which is 
equivalent to $\V_k$ only having a small number of nonzero elements.

\vspace{-2mm}
\subsection{CUR matrix decompositions}
\label{sec:cur}
\vspace{-1mm}

CUR decompositions were proposed by Drineas and
Mahoney~\cite{CUR_PNAS,DMM08_CURtheory_JRNL} to provide a 
low-rank
approximation to a data matrix $\X$ by using only a small number of
actual columns and/or rows of $\X$.
Fast randomized variants~\cite{dkm_matrix3}, deterministic
variants~\cite{GT01}, Nystr\"{o}m-based variants~\cite{BW07_WKSHP, KMT09b}, and
heuristic variants~\cite{SXZF07_cmdmatrix} have also been considered.
Observing that the best rank-$k$ approximation to the SVD provides the best
set of $k$ linear combinations of all the columns, one can ask for the best
set of $k$ \emph{actual} columns.
Most formalizations of ``best'' lead to intractable combinatorial
optimization problems~\cite{CUR_PNAS}, but one can take advantage of
oversampling (choosing slightly more than $k$ columns) and randomness as
computational resources to obtain strong quality-of-approximation
guarantees.

\begin{thm}[Relative-error CUR~\cite{CUR_PNAS}]
\label{thm:CUR}
Given an arbitrary matrix $\X\in\r^{n\times p}$ and an integer
$k$,
there exists a randomized algorithm that chooses a random subset
$\I\subset\{1,\ldots,p\}$ of size
$c=O(k\log k\log(1/\delta)/\epsilon^2)$ such that
$\X^{\I}$, the $n\times c$
submatrix containing those $c$ columns of $\X$, satisfies
  \begin{align}
  \label{eq:err}
    ||\X-\X^{\I}\X^{\I+}\X||_F
    =    \min_{\B\in\r^{c\times p}}||\X-\X^\I\B||_F
    \leq (1+\epsilon)||\X-\X_k||_F,
  \end{align}
with probability at least $1-\delta$, where $\X_k$ is the best rank $k$ approximation to $\X$.
\end{thm}

\noindent
The algorithm referred to by Theorem~\ref{thm:CUR} is very simple:
\begin{compactenum}[1)]
\item
Compute the \emph{normalized statistical leverage scores}, defined below in
(\ref{eqn:col_probs}).
\item
Form $\I$ by randomly sampling  $c$
columns of $\X$, using these normalized statistical leverage scores as an importance sampling distribution.
\item
Return the $n \times c$ matrix $\X^\I$ consisting of these selected
columns.
\end{compactenum}
The key issue here is the choice of
the importance sampling distribution.
Let the $p \times k$ matrix $\V_k$ be
the top-$k$ right singular vectors of $\X$.
Then the \emph{normalized statistical leverage scores} are
\begin{equation}
\label{eqn:col_probs}
\pi_i %= \frac{1}{k}\sum_{j=1}^k (v^j_i)^2
      = \frac{1}{k} ||\V_{k(i)}||_2^2,
\end{equation}
for all $i=1,\ldots,p$, where $\V_{k(i)}$ denotes the $i$-th row of $\V_k$.
These scores, proportional to the Euclidean norms of the \emph{rows} of
the top-$k$ right singular vectors,
define the relevant nonuniformity structure to be used to identify
good (in the sense of Theorem~\ref{thm:CUR}) columns.
In addition, these scores are proportional to the diagonal elements of the
projection matrix onto the top-$k$ right singular subspace.
Thus, they generalize the so-called hat matrix~\cite{HW78}, and they
have a natural interpretation as capturing the ``statistical leverage'' or
``influence'' of a given column on the best low-rank fit of the data 
matrix~\cite{HW78,CUR_PNAS}.

\vspace{-2mm}
\subsection{Regularized sparse PCA methods}
\label{sec:regul-sparse-pca}
\vspace{-1mm}

SPCA methods attempt to make PCA easier to interpret
for domain experts by finding sparse approximations to the \emph{columns} of
$\V$.\footnote{For SPCA, we only consider sparsity in the right singular vectors $\V$ and not in the left singular vectors $\U$.  This is similar to considering only the choice of columns and not of both columns and rows in CUR.}
There are several variants of SPCA.
For example, Jolliffe \emph{et al.}~\cite{Jolliffe} and Witten
\emph{et al.}~\cite{Witten} use the maximum variance interpretation of PCA
and provide an optimization problem which explicitly encourages
sparsity in $\V$ based on a Lasso constraint~\cite{Tib96}.
d'Aspremont \emph{et al.}~\cite{AGJL07} take a similar approach, but
instead formulate the problem as an SDP.

Zou \emph{et al.}~\cite{ZHT06} use the minimum reconstruction error
interpretation of PCA to suggest a different approach to the SPCA problem;
this formulation will be most relevant to our present purpose.
They begin by formulating
PCA as the solution to a regression-type problem.

\begin{thm}[Zou \emph{et al.}~\cite{ZHT06}]
\label{thm:PCA-reg}
Given an arbitrary matrix $\X\in\r^{n\times p}$ and an integer $k$,
let $\A$ and $\W$ be $p \times k$  matrices.
Then, for any $\lambda > 0$, let
\begin{align}
\label{eq:Zou1}
(\A^*,\V^*_k)=\argmin_{\A,\W\in \r^{p\times k}}||\X-\X\W\A^T||_F^2 + \lambda||\W||_F^2 \quad\st\A^T\A=\II_k .
\end{align}
Then, the minimizing matrices $\A^*$ and $\V^*_k$ satisfy
$\A^{*(i)} = s_i \V^{(i)}$ and $\V^{*(i)}_k = s_i \frac{\bSigma_{ii}^2}{\bSigma_{ii}^2+\lambda}\V^{(i)}$,
where $s_i=1$ or $-1$.
\end{thm}
\noindent
That is, up to signs, $\A^*$ consists of the top-$k$ right singular vectors
of $\X$, and $\V^*_k$  consists of those same vectors ``shrunk'' by a factor
depending on the corresponding singular value.
Given this regression-type characterization of PCA, Zou
\emph{et al.}~\cite{ZHT06} then ``sparsify'' the formulation by adding an
$L_1$ penalty on~$\W$:
\begin{align}
\label{eq:Zou2}
(\A^*,\V^*_k)=\argmin_{\A,\W\in \r^{p\times k}}||\X-\X\W\A^T||_F^2 + \lambda||\W||_F^2+\lambda_1||\W||_1 \quad\st\A^T\A=\II_k ,
\end{align}
where  $||\W||_1 = \sum_{ij}|\W_{ij}|$.
This regularization tends to sparsify $\W$ element-wise, so that the
solution $\V^*_k$ gives a sparse approximation of $\V_k$.

\vspace{-2mm}
\section{Expressing CUR as an optimization problem}
\label{sec:cur-optim-fram}
\vspace{-1mm}

In this section, we present an optimization formulation of CUR.
Recall, from Section~\ref{sec:cur}, that CUR takes a purely algorithmic 
approach to the problem of approximating a matrix in terms of a small 
number of its columns.  
That is, it achieves sparsity indirectly by randomly selecting $c$ columns, 
and it does so in such a way that the reconstruction error is small with 
high probability (Theorem \ref{thm:CUR}).  
By contrast, SPCA methods are generally formulated as the exact solution to 
an optimization problem.  

From Theorem \ref{thm:CUR}, it is clear that CUR seeks a
subset $\I$ of size $c$ for which
$\min_{\B\in\r^{c\times p}}||\X-\X^\I\B||_F$
is small.  In this sense, CUR can be viewed as a
randomized algorithm for approximately solving the following combinatorial optimization problem:
\begin{align}\label{eq:allsubsets}
\min_{\I\subset\{1,\ldots,p\}}\min_{\B\in\r^{c\times p}}||\X-\X^\I\B||_F\quad
  \st|\I|\le c.
\end{align}
In words, this objective asks for the subset of $c$ columns of $\X$ which best
describes the entire matrix $\X$.  Notice that relaxing  $|\I|=c$ to
$|\I|\le c$ does  not affect the optimum.  This optimization problem is analogous
to all-subsets multivariate regression \cite{hast-tibs-fried}, which is known to be NP-hard.

However, by using ideas from the optimization literature we can 
approximate this combinatorial problem as a regularized
regression problem that is convex.  First, notice that
\eqref{eq:allsubsets} is
equivalent to
\begin{align}\label{eq:allsubsets2}
  \min_{\B\in\r^{p\times p}}||\X-\X\B||_F
  \quad\st\sum_{i=1}^p\one_{\{||\B_{(i)}||_2\neq 0\}}\le c ,
\end{align}
where we now optimize over a $p\times p$ matrix $\B$. To see the
equivalence between \eqref{eq:allsubsets} and \eqref{eq:allsubsets2},
note that the constraint in \eqref{eq:allsubsets2} is the same as
finding some subset $\I$ with $|\I|\le c$ such that $\B_{\I^c}=\0$.

The formulation in \eqref{eq:allsubsets2} provides a natural entry point to
proposing a convex optimization approach corresponding to CUR. First notice
that \eqref{eq:allsubsets2} uses an $L_0$ norm on the rows
of $\B$, which is not convex. However, we can approximate the $L_0$ constraint
by a \textit{group lasso} penalty, which uses a well-known convex
heuristic proposed
by Yuan \emph{et al.}~\cite{Grouplasso} that encourages prespecified \textit{groups} of parameters to be
simultaneously sparse.
Thus, the combinatorial problem in
\eqref{eq:allsubsets} can be approximated by the following
convex (and thus tractable) problem:
\begin{prob}[\bf Group lasso regression: \glreg]
\label{prob:glreg}
Given an arbitrary matrix $\X\in\r^{n\times p}$,
let $\B\in\r^{p\times p}$ and $t>0$.
The \glreg~ problem is to solve
\begin{align}\label{eq:regr_gplasso}
  \B^* = \argmin_{\B}||\X-\X\B||_F
  \quad\st\sum_{i=1}^p||\B_{(i)}||_2\le t,
\end{align}
where $t$ is chosen to get $c$ nonzero rows in $\B^*$.
\end{prob}
Since the rows of $\B$ are grouped together in the penalty
$\sum_{i=1}^p||\B_{(i)}||_2$, the row vector $\B_{(i)}$ will
tend to be either dense or entirely zero.
Note also that the algorithm to solve Problem~\ref{prob:glreg} is a special 
case of Algorithm~\ref{alg:groupLasso} (see below), which solves the 
\glpca~ problem, to be introduced later.
(Finally, as a side remark, note that our proposed \glreg~ is strikingly 
similar to a recently proposed method for sparse inverse covariance 
estimation~\cite{Friedman10,Peng09}.)

\vspace{-2mm}
\section{Distinguishing CUR from SPCA}
\label{sec:connections-pca}
\vspace{-1mm}

Our original intention in casting CUR in the optimization framework was to 
understand better whether CUR could be seen as an SPCA-type method.  
So far, we have established CUR's connection to regression by showing that 
CUR can be thought of as an approximation algorithm for the sparse 
regression problem~\eqref{eq:allsubsets2}.  
In this section, we discuss the relationship between regression and PCA, and
we show that CUR cannot be directly cast as an SPCA method.

To do this, recall that 
regression, in particular ``self'' regression, finds a $\B\in\r^{p\times
p}$ that minimizes
\begin{align}\label{eq:reconCUR}
  ||\X-\X\B||_F.
\end{align}
On the other hand, PCA-type methods find a set of directions $\W$ that 
minimize
\begin{align}\label{eq:PCArecon}
  \Err(\W):=||\X-\X\W\W^+||_F.
\end{align}
Here, unlike in \eqref{eq:PCA}, we do not assume that $\W$ is
orthogonal, since the minimizer produced from SPCA methods is often not
required to be orthogonal (recall Section \ref{sec:regul-sparse-pca}).

Clearly, with no constraints on $\B$ or $\W$, we can trivially achieve
zero reconstruction error in both cases by taking $\B=\II_p$ and $\W$ any $p\times p$
full-rank matrix.  However, with additional constraints, these two
problems can be very different. It is common to consider
sparsity and/or rank constraints. We have seen in Section
\ref{sec:cur-optim-fram} that CUR effectively requires $\B$ to
be row-sparse; in the standard PCA setting, $\W$ is taken to be
rank $k$ (with $k<p$), in which case \eqref{eq:PCArecon} is
minimized by $\V_k$ and obtains the optimal value
$\Err(\V_k)=||\X-\X_k||_F$; finally, for SPCA, $\W$ is further required to be
sparse.

To illustrate the difference between the reconstruction errors
\eqref{eq:reconCUR} and \eqref{eq:PCArecon} when extra constraints are
imposed, consider the 2-dimensional toy example in Figure~\ref{fig:recon}.
In this example,  
we compare regression with a row-sparsity
constraint to PCA with both rank and sparsity constraints.  With $\X\in\r^{n\times 2}$, we plot $\X^{(2)}$ against $\X^{(1)}$ as the solid points in both
plots of Figure~\ref{fig:recon}. Constraining $\B_{(2)}=0$ (giving
row-sparsity, as with CUR methods),
\eqref{eq:reconCUR} becomes
$\min_{B_{12}}||\X^{(2)}-\X^{(1)}B_{12}||_2$, which is a simple linear
regression, represented by the black thick line and minimizing the sum
of squared vertical errors as shown. The red line (left plot) shows the
first principal component direction, which minimizes $\Err(\W)$ among all rank-one matrices
$\W$.  Here, $\Err(\W)$ is the sum of squared projection distances (red
dotted lines). Finally, if $\W$ is further required to
be sparse in the $\X^{(2)}$ direction (as with SPCA methods), we get
the rank-one, sparse projection represented by the green line in
Figure~\ref{fig:recon} (right). The two sets of dotted lines in each
plot clearly differ, indicating that their corresponding
reconstruction errors are different as well.
Since we have shown that CUR is minimizing a regression-based
objective, this toy example suggests that CUR may not in fact be
optimizing a PCA-type objective such as \eqref{eq:PCArecon}.
Next, we will make this intuition more precise.

\begin{figure}[!t]
  \centering
  \psfrag{Regression}[c][c]{\tiny Regression}
  \psfrag{PCA}[c][c]{\tiny PCA}
  \psfrag{SPCA}[c][c]{\tiny SPCA}
  \psfrag{X1}[c][c]{\scriptsize $\X^{(1)}$}
  \psfrag{X2}[c][c]{\scriptsize $\X^{(2)}$}
  \psfrag{ErrReg}[c][c]{\scriptsize ~~error~\eqref{eq:reconCUR}}
  \psfrag{ErrPCA}[c][c]{\scriptsize ~~error~\eqref{eq:PCArecon}}
  \includegraphics*[width=0.45\textwidth, viewport=18 50 480 380, clip]{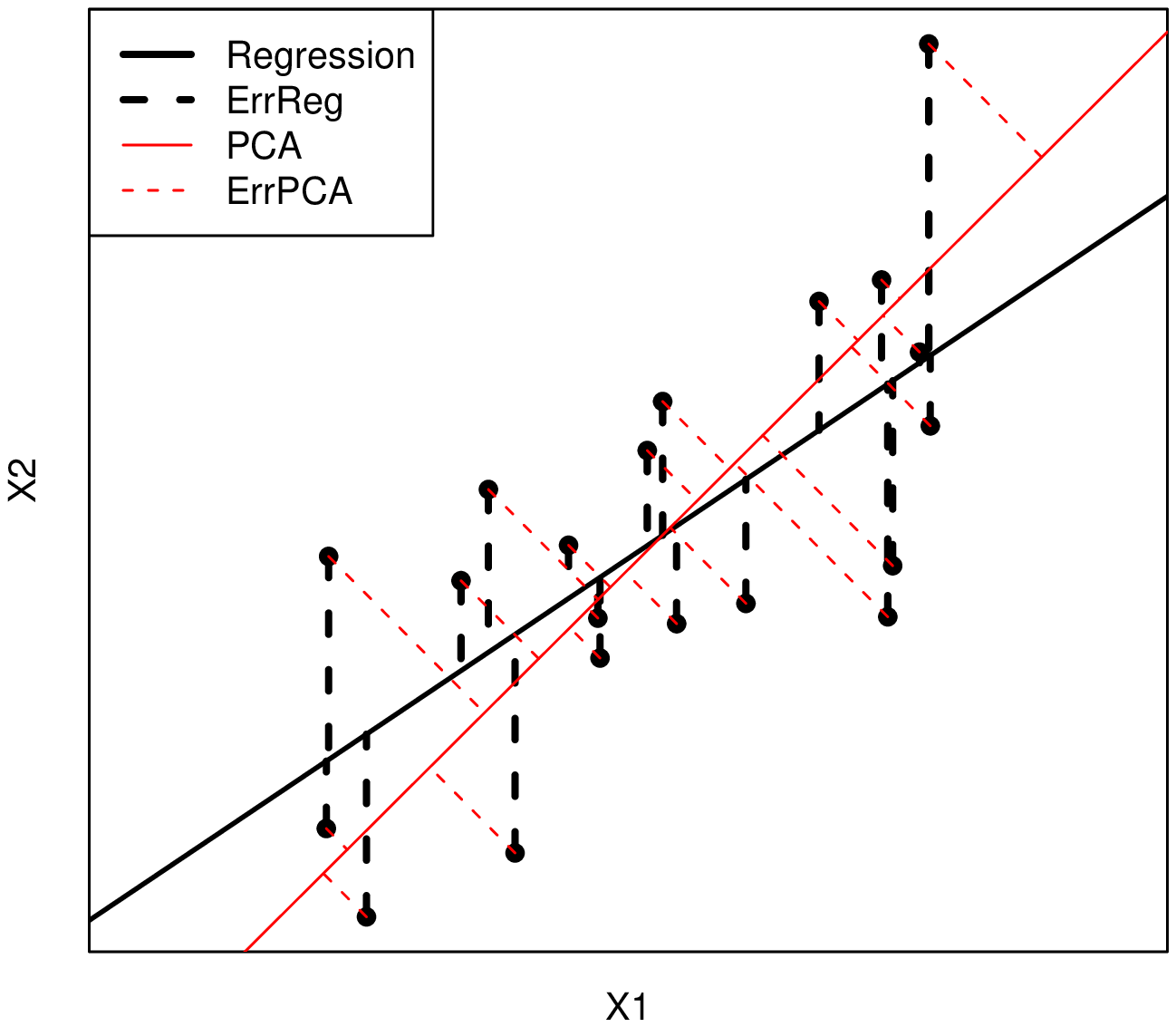}
  \includegraphics*[width=0.45\textwidth, viewport=18 50 480 380,
  clip]{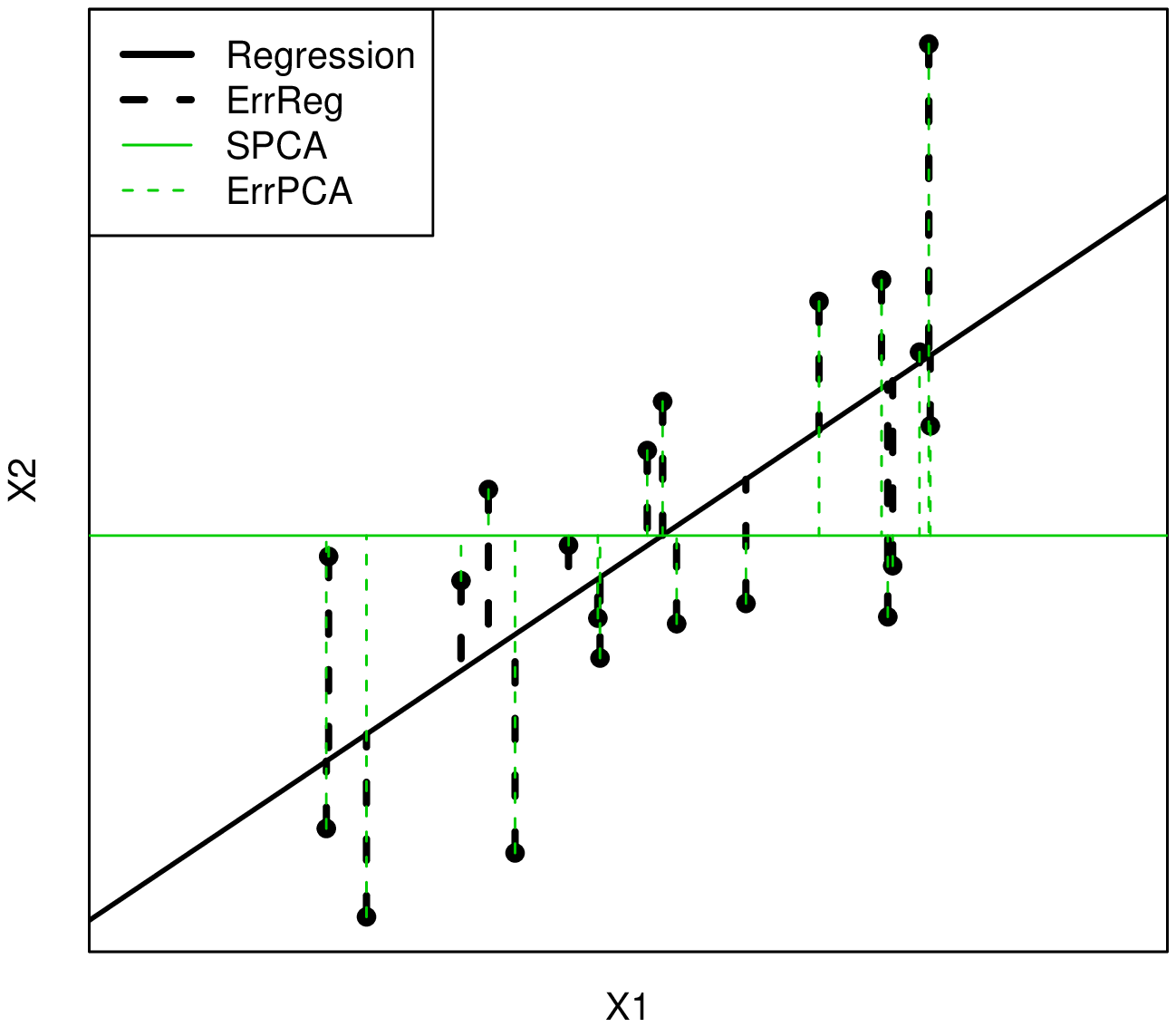}
\caption{ Example of the difference in reconstruction errors
    \eqref{eq:reconCUR} and \eqref{eq:PCArecon}, when additional constraints
imposed.
Left: regression with row-sparsity constraint (black)
compared with PCA with low rank constraint (red). Right: regression
with row-sparsity constraint (black)
compared with PCA with low rank and sparsity constraint (green). In both
plots, the corresponding errors are represented by the
dotted lines.}
  \label{fig:recon}
\end{figure}

The first step to showing that CUR
is an SPCA method would be to produce a matrix $\Vcur$ for which
$\X^{\I}\X^{\I+}\X=\X\Vcur\Vcur^{+}$, \emph{i.e.} to express CUR's
approximation in the form of an SPCA approximation.  However, this equality implies
$\Lcol(\X\Vcur\Vcur^{+})\subseteq \Lcol(\X^{\I})$, meaning that $(\V_{\textsc{cur}})_{\I^c}=\0$.
If such a $\Vcur$ existed, then clearly $\Err(\Vcur) =
||\X-\X^{\I}\X^{\I+}\X||_F$, and so CUR could be regarded as
implicitly performing sparse PCA in the sense that (a) $\Vcur$ is
sparse; and (b) by Theorem \ref{thm:CUR} (with high
probability), $\Err(\Vcur)\le (1+\epsilon)\Err(\V_k)$.  Thus, the existence of such a $\Vcur$ would cast CUR
directly as a randomized approximation algorithm for SPCA.
However, the following theorem states that unless an unrealistic
constraint on $\X$ holds, there does not exist a matrix $\Vcur$ for
which $\Err(\Vcur) =||\X-\X^{\I}\X^{\I+}\X||_F$.
The larger implication of this theorem is that CUR cannot be directly
viewed as an SPCA-type method.
\begin{thm}
  Let $\I\subset\{1,\ldots,p\}$ be an index set and suppose
  $\W\in\r^{p\times p}$ satisfies $\W_{\I^c}=\0$. Then,
  $$||\X-\X\W\W^+||_F > ||\X-\X^\I\X^{\I+}\X||_F,$$
  unless $\mathcal L_{\mathrm{col}}(\X^\I)\perp\mathcal
  L_{\mathrm{col}}(\X^{\I^c})$, in which case ``$\ge$'' holds.
\label{thm:cur-not-pca}
\end{thm}
\begin{proof}
  \begin{align*}
    ||\X-\X\W\W^+||_F^2&=||\X-\X^{\I}\W_{\I}\W^+||_F^2 =||\X-\X^{\I}\W_{\I}(\W_\I^T\W_\I)^{-1}\W^T||_F^2\\
    &=||\X^{\I}-\X^{\I}\W_{\I}\W_{\I}^+||_F^2+||\X^{\I^c}||_F^2
    \ge||\X^{\I^c}||_F^2\\
    &=||\X^{\I^c}-\X^\I\X^{\I+}\X^{\I^c}||_F^2+||\X^\I\X^{\I+}\X^{\I^c}||_F^2\\
    &=||\X-\X^\I\X^{\I+}\X||_F^2 + ||\X^\I\X^{\I+}\X^{\I^c}||_F^2
    \ge||\X-\X^\I\X^{\I+}\X||_F^2.
  \end{align*}
  The last inequality is strict unless $\X^\I\X^{\I+}\X^{\I^c}=\0$.
\end{proof}

\vspace{-5mm}
\section{CUR-type sparsity and the group lasso SPCA}
\label{sec:group-lasso-pca}
\vspace{-2mm}

Although CUR cannot be directly cast as an SPCA-type method, in this 
section we propose a sparse PCA approach (which we call the group lasso 
SPCA or \glpca) that accomplishes something very close to CUR. 
Our proposal produces a $\V^*$ that has rows that are entirely zero, and it 
is motivated by the following two observations about CUR.
First, following from the definition of the leverage scores
\eqref{eqn:col_probs}, CUR chooses columns of $\X$ based on the norm of their
corresponding rows of $\V_k$. 
Thus, it essentially ``zeros-out'' the rows of $\V_k$ with small norms (in 
a probabilistic sense). 
Second, as we have noted in Section~\ref{sec:connections-pca}, if CUR could 
be expressed as a PCA method, its principal directions matrix ``$\Vcur$'' 
would have $p-c$ rows that are entirely zero, corresponding to removing 
those columns of $\X$.

Recall that Zou \emph{et al.}~\cite{ZHT06} obtain a sparse  $\V^*$ by
including in~(\ref{eq:Zou2}) an additional $L_1$ penalty from the
optimization problem~\eqref{eq:Zou1}. Since the $L_1$ penalty is on
the entire matrix viewed as a vector, it encourages only unstructured sparsity.
To achieve the CUR-type row sparsity, we propose
the following modification of~\eqref{eq:Zou1}:
\begin{prob}[\bf Group lasso SPCA: \glpca]
\label{prob:cur-spca}
Given an arbitrary matrix $\X\in\r^{n\times p}$ and an integer $k$,
let $\A$ and $\W$ be $p \times k$  matrices, and let $\lambda,\lambda_1>0$.
The \glpca~ problem is to solve
\begin{align}
\label{eq:groupLasso}
(\A^*,\V^*) =
  \argmin_{\A,\W}||\X-\X\W\A^T||_F^2 + \lambda||\W||_F^2 +
  \lambda_1\sum_{i=1}^p||\W_{(i)}||_2 ~\st\A^T\A=\II_k.
\end{align}
\end{prob}
Thus, the lasso penalty $\lambda_1||\W||_1$ in~\eqref{eq:Zou2} is replaced
in~\eqref{eq:groupLasso} by a group lasso penalty
$\lambda_1\sum_{i=1}^p||\W_{(i)}||_2$, where rows of $\W$ are
grouped together so that each row of $\V^*$ will tend to be either
dense or entirely zero.

Importantly, the \glpca~ problem is not convex in $\W$ and $\A$ together; 
it is, however, convex in $\W$, and it is easy to solve in $\A$.
Thus, analogous to the treatment in Zou \emph{et al.}~\cite{ZHT06}, we 
propose an iterative alternate-minimization algorithm to solve \glpca.
This is described in Algorithm~\ref{alg:groupLasso}; 
and the justification of this algorithm is given in Section
\ref{sec:deriv-alg}. Note that if we fix $\A$ to be $\II$ throughout, then
Algorithm~\ref{alg:groupLasso} can be used to solve the \glreg~ problem 
discussed in Section~\ref{sec:cur-optim-fram}.

We remark that such row-sparsity in $\V^*$
can have either advantages or disadvantages.
Consider, for example, when there are a small
number of informative columns in $\X$ and the rest are not
important for the task at hand~\cite{CUR_PNAS,Paschou07b}.
In such a case, we would expect that enforcing entire rows to be zero would lead to better identification of the signal columns;
and this has been empirically observed in the application
of CUR to DNA SNP analysis~\cite{Paschou07b}.
The unstructured $\V^*$, by contrast, would not be able to ``borrow
strength'' across all columns of $\V^*$ to differentiate the signal columns from
the noise columns.
On the other hand, requiring such structured sparsity is more
restrictive and may not be desirable.
For example, in microarray analysis in which we have measured $p$ genes on
$n$ patients, our goal may be to find several underlying factors.
Biologists have identified ``pathways'' of interconnected genes~\cite{Subramanian05}, and it
would be desirable if each sparse factor could be identified with a
different pathway (that is, a different set of genes).
Requiring all factors of $\V^*$ to exclude the same $p-c$ genes does
not allow a different sparse subset of genes to be active in each factor.

We finish this section by pointing out that
while most SPCA methods only enforce unstructured zeros in $\V^*$, the
idea of having a structured sparsity in the
PCA context has very recently been explored~\cite{JOB09_TR}.
Our \glpca~ problem falls within the broad framework of this
idea.

\begin{algorithm}[t]
\label{alg:groupLasso}
\caption{Iterative algorithm for solving the \glpca~ (and \glreg)~problems.\\
$\text{\hspace{20mm}}$(For the \glreg~ problem, fix $\A=\II$ throughout this algorithm.)}

\KwIn{Data matrix $\X$ and initial estimates for $\A$ and $\W$}
\KwOut{Final estimates for $\A$ and $\W$}
\Repeat{convergence}{
\nl Compute SVD of $\X^T\X\W$ as $\U\D\V^T$ and then $\A\leftarrow \U\V^T$\; \label{step:AgivenB}
    $\mathcal{S}\leftarrow \{i:\; ||\W_{(i)}||_2\neq 0 \}$\;

 \For{$i\in\mathcal{S}$}{
\nl      Compute $\bb_i=\sum_{j\neq i}\left(\X^{(j)T}\X^{(i)}\right)\W_{(j)}^T$\;

         \eIf{     $||\A^T\X^T\X^{(i)}-\bb_i||_2\leq\lambda_1/2$ }
             { \nl $\W_{(i)}^T\leftarrow \0$\;\label{step:if} }
             { \nl $\W_{(i)}^T\leftarrow\frac{2}{2||\X^{(i)}||_2^2+\lambda+\lambda_1/||\W_{(i)}||_2}\left(\A^T\X^T\X^{(i)}-\bb_i\right)$\;\label{step:else} }
}
}
\end{algorithm}

\vspace{-3mm}
\section{Empirical Comparisons}
\label{sec:empir-comp}
\vspace{-2mm}

In this section, we evaluate the performance of the four
methods discussed above on both synthetic and real data.  In
particular, we compare the randomized CUR algorithm of Mahoney and Drineas~\cite{CUR_PNAS,DMM08_CURtheory_JRNL} to our \glreg~ (of Problem~\ref{prob:glreg}),
and we compare the SPCA algorithm proposed by Zou \emph{et al.}~\cite{ZHT06} to
our \glpca~ (of Problem~\ref{prob:cur-spca}).  We have also compared
 against the  SPCA algorithm of Witten \emph{et al.}~\cite{Witten}, and we
 found the results to be very similar to those of Zou \emph{et al.}

\vspace{-3mm}
\subsection{Simulations}
\label{sec:simulations}
\vspace{-2mm}

We first consider synthetic examples of the form $\X=\wh\X + \mathbf{E},$
where $\wh\X$ is the underlying signal matrix and $\mathbf{E}$ is a matrix of
noise.  In all our simulations, $\mathbf{E}$ has i.i.d. $\n(0,1)$
entries, while the signal $\wh\X$ has one of the following forms:
\begin{enumerate}
\item[Case I)] $\wh \X=[\mathbf{0}_{n\times(p-c)};\wh \X^*]$ where the
  $n\times c$ matrix $\wh\X^*$ is the nonzero part of $\wh\X$. In
  other words, $\wh\X$ has $c$ nonzero columns and does not
  necessarily have a low-rank structure.
\item[Case II)] $\wh\X = \U\V^T$ where $\U$ and
$\V$ each consist of $k<p$ orthogonal columns.
In addition to being low-rank, $\V$ has entire rows equal to zero
(\emph{i.e.} it is row-sparse).
\item[Case III)] $\wh\X = \U\V^T$ where $\U$ and
$\V$ each consist of $k<p$ orthogonal columns.  Here $\V$ is low-rank
and sparse, but
the sparsity is not structured (\emph{i.e.} it is scattered-sparse).
\end{enumerate}

A successful method attains low reconstruction error of the true
signal $\wh\X$ and has high precision
in identifying correctly the zeros in the underlying model.  As
previously discussed, the four methods optimize for different types of
reconstruction error.  Thus, in comparing CUR and \glreg, we
use the regression-type reconstruction error
$\Err_\mathrm{reg}(\I) = ||\wh\X - \X^\I\X^{\I+}\X||_F$, whereas for the comparison of
SPCA and \glpca, we use the PCA-type error
$\Err(\V)=||\wh\X-\X\V\V^+||_F.$

Table~\ref{tab:simulation} presents the simulation results from the
three cases. All
comparisons use $n=100$ and $p=1000$. In Case II and III, the
signal matrix has rank $k=10$. The underlying sparsity level is $20\%$,
\emph{i.e.} $80\%$ of the entries of $\wh\X$ (Case I) and $\V$ (Case
II\&III) are zeros. Note that all methods except for
\glreg~ require the rank $k$ as an input, and we always take it to
be 10 even in Case I.
For easy comparison, we have tuned each method to have  the correct total
number of zeros.
The results are averaged
over 5 trials.

\begin{table}[!h]
  \centering
\begin{tabular}{lllll}
  \toprule
  &Methods    &  Case I & Case II & Case III \\
  \midrule
  \multirow{2}{*}{$\Err_\mathrm{reg}(\I)$} &  CUR & 316.29 (0.835) &   315.28 (0.797) & 315.64 (0.166)\\
                    &\glreg &   316.29 (0.989)    &  315.28 (0.750)  &  315.64 (0.107)\\
  \midrule
  \multirow{2}{*}{$\Err(\V)$}&SPCA &  177.92 (0.809)& 44.388 (0.799) &44.995 (0.792)\\
  &\glpca &   141.85 (0.998)   &  37.310 (0.767)    & 45.500 (0.804)\\
  \bottomrule
\end{tabular}
\caption{Simulation results:~The reconstruction errors and the
percentages of correctly identified zeros (in parentheses).}
  \label{tab:simulation}
\end{table}

We notice in Table~\ref{tab:simulation} that the two regression-type
methods CUR and \glreg~ have very similar performance. As we
would expect, since CUR only uses information in the top $k$ singular
vectors, it does slightly worse than \glreg~ in terms of
precision when the underlying signal is not low-rank (Case I). 
In addition, both methods
perform poorly if the sparsity is not structured as in Case III.
The two PCA-type methods perform similarly as well. Again, the group
lasso method seems to work better in Case I. 
We note 
that the precisions reported here are based on
element-wise sparsity---if we were measuring row-sparsity, methods like
SPCA would perform poorly since they do not encourage entire rows to be zero.

\vspace{-3mm}
\subsection{Microarray example}
\label{sec:microarray-example}
\vspace{-2mm}

We next consider a microarray dataset of soft tissue tumors
studied by Nielsen \emph{et al.}~\cite{LancetCancer02}.
Mahoney and Drineas~\cite{CUR_PNAS} apply CUR to this dataset  of $n = 31$
tissue samples and $p = 5520$ genes.  As with the simulation results, we use two sets of 
comparisons: we compare CUR with \glreg, and we compare SPCA
with \glpca.  Since  we do not observe the underlying truth
$\wh \X$, we take $\Err_\mathrm{reg}(\I) = ||\X - \X^\I\X^{\I+}\X||_F$
and $\Err(\V)=||\X-\X\V\V^+||_F.$  Also, since we do not observe the
true sparsity, we cannot measure the precision as we do in Table
\ref{tab:simulation}.  The left plot in Figure \ref{fig:pnas_compare} shows
$\Err_\mathrm{reg}(\I)$ as a function of $|\I|$.  We see that CUR and
\glreg~ perform similarly.  (However, since CUR is a randomized
algorithm, on every run it gives a different result.  From a practical
standpoint, this feature of
CUR can be disconcerting to biologists wanting to report a single set
of important genes.  In this light, \glreg~ may be thought of as an 
attractive non-randomized alternative to~CUR.)
The right plot of Figure \ref{fig:pnas_compare} compares
\glpca~ to SPCA (specifically, Zou \emph{et
  al.}~\cite{ZHT06}).  Since SPCA does not explicitly enforce
row-sparsity, for a
gene to be not used in the model requires \emph{all} of the ($k=4$) columns
of $\V^*$ to exclude it.  This likely explains the advantage of \glpca~ over 
SPCA seen in the figure.

\begin{figure}
  \centering
  \psfrag{err1}[c][c]{\scriptsize $\Err_\mathrm{reg}(\I)$}
  \psfrag{err2}[c][c]{\scriptsize $\Err(\V)$}
  \psfrag{NumberOfGenes}[c][c]{\scriptsize Number of genes used}
  \psfrag{dataset}[c][c]{\scriptsize \bf Microarray Dataset}
  \psfrag{GL-Reg}[c][c]{\tiny \glreg}
  \psfrag{GL-PCA}[c][c]{\tiny \glpca}
  \psfrag{CUR}[c][c]{\tiny CUR}
  \psfrag{SPCA}[c][c]{\tiny SPCA}
  \includegraphics*[width=0.45\textwidth, viewport=0 30 480 420, clip]{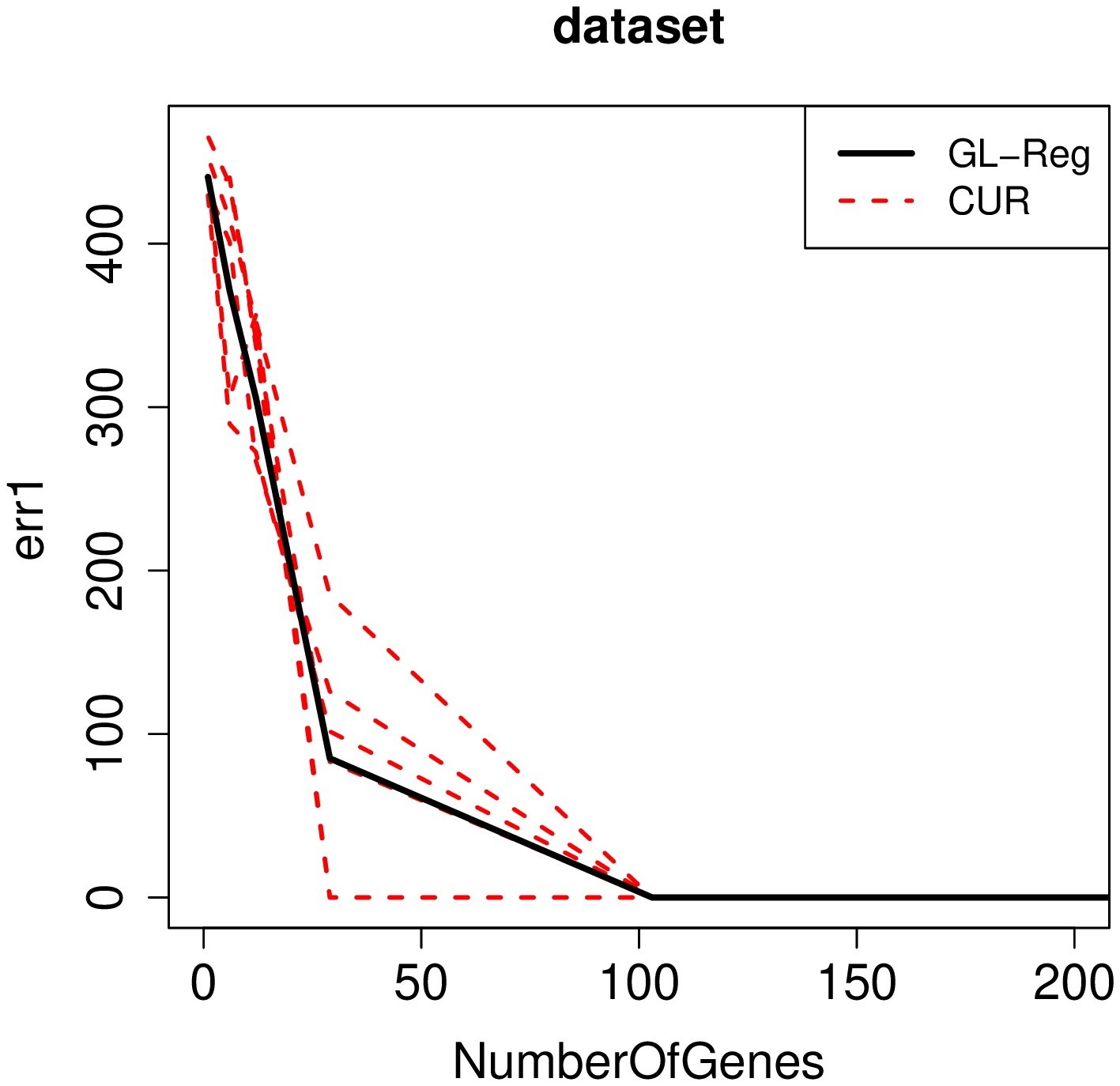}
  \includegraphics*[width=0.45\textwidth, viewport=0 30 480 420, clip]{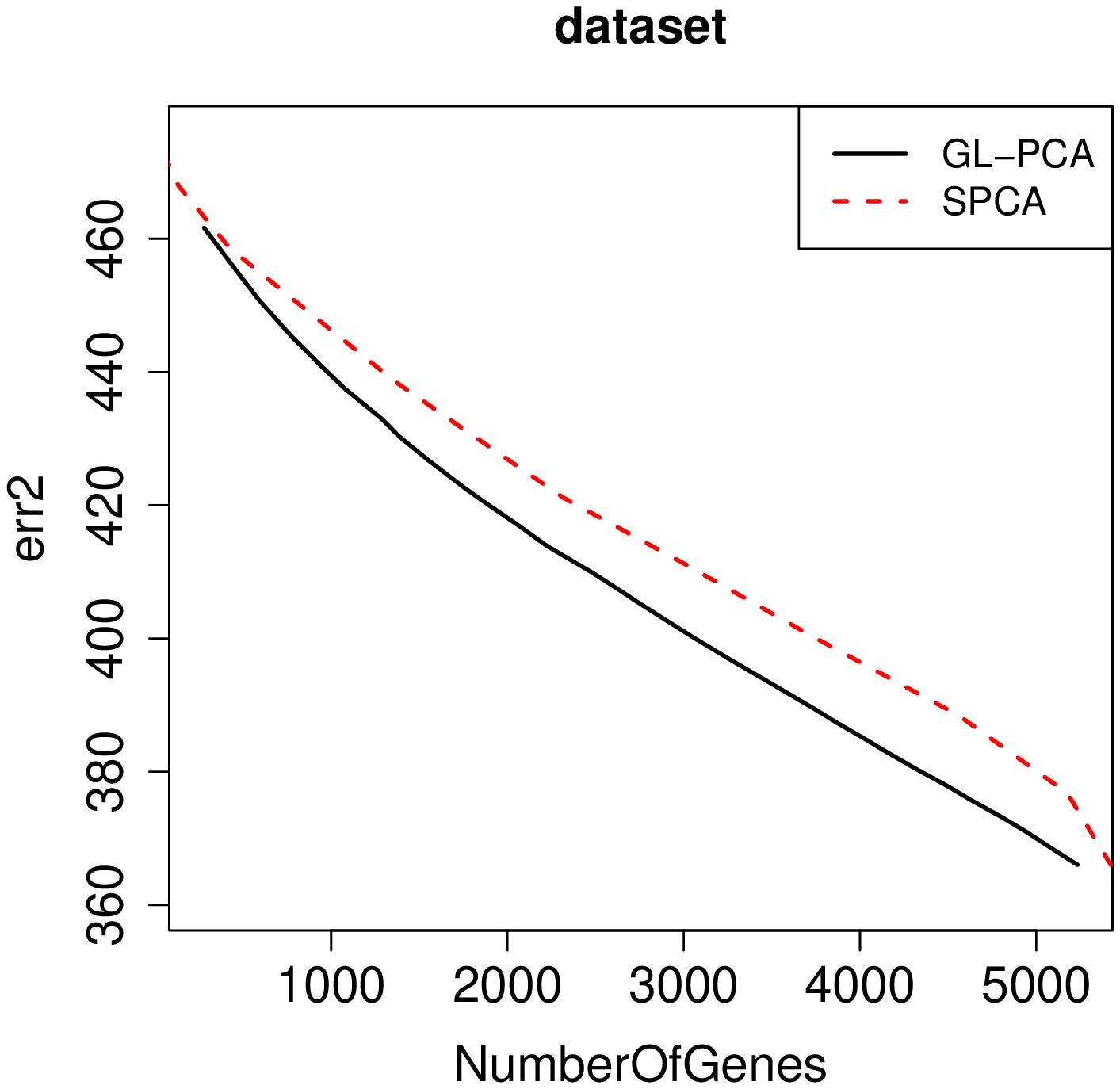}
  \caption{Left: Comparison of CUR, multiple runs, with \glreg; Right:
    Comparison of \glpca~ with SPCA (specifically, Zou \emph{et al.}~\cite{ZHT06}).}
  \label{fig:pnas_compare}
\end{figure}

\vspace{-4mm}
\section{Justification of Algorithm~\ref{alg:groupLasso}}
\label{sec:deriv-alg}
\vspace{-2mm}

The algorithm alternates between minimizing with respect to $\A$ and
$\B$ until convergence.

\textbf{Solving for $\A$ given $\B$:~} If $\B$ is fixed, then the regularization penalty
in~\eqref{eq:groupLasso} can be ignored, in which case the optimization
problem becomes $\min_{\A} ||\X-\X\B\A^T||_F^2$ subject to $\A^T\A=I$.
This problem was considered by Zou \emph{et al.}~\cite{ZHT06}, who showed
that the solution is obtained by computing the SVD of $(\X^T\X)\B$ as
$(\X^T\X)\B=\U\D\V^T$ and then setting $\wh{\A} = \U\V^T$.
This explains step~\ref{step:AgivenB} in Algorithm~\ref{alg:groupLasso}.

\textbf{Solving for $\B$ given $\A$:}
If $\A$ is fixed, then~\eqref{eq:groupLasso} becomes an unconstrained
convex optimization problem
in $\B$.
The subgradient equations (using that $\A^T\A=\II_k$) are
\begin{align}\label{eq:subgradient}
2\B^T\X^T\X^{(i)}-2\A^T\X^T\X^{(i)}+2\lambda
\B_{(i)}^T+\lambda_1\s_i=\0;\quad i=1,\dots,p,
\end{align}
where the subgradient vectors $\s_i=\B_{(i)}^T/||\B_{(i)}||_2$ if
$\B_{(i)}\neq \0$, or $||\s_i||_2 \le 1$ if $\B_{(i)}=\0$.
Let us define
$
\bb_i = \sum_{j\neq i}(\X^{(j)T}\X^{(i)})\B_{(j)}^T
    = \B^T\X^T\X^{(i)} -||\X^{(i)}||_2^2\B^T_{(i)},
$
so that the subgradient equations can be written as
\begin{align}\label{eq:subgradient2}
\bb_i+(||\X^{(i)}||_2^2+\lambda)\B_{(i)}^T-\A^T\X^T\X^{(i)}+(\lambda_1/2)\s_i=\0.
\end{align}

The following claim explains Step~\ref{step:if} in
Algorithm~\ref{alg:groupLasso}.
\begin{claim}\label{claim1}
  $\B_{(i)}=\0$ if and only if $||\A^T\X^T\X^{(i)}-\bb_i||_2\leq\lambda_1/2$.
\end{claim}
\begin{proof}
  First, if $\B_{(i)}=\0$, the subgradient
  equations~\eqref{eq:subgradient2} become
  $\bb_i-\A^T\X^T\X^{(i)}+(\lambda_1/2)\s_i=\0$. Since $||\s_i||_2 \le 1$ if $\B_{(i)}=\0$,
  we have $||\A^T\X^T\X^{(i)}-\bb_i||_2\leq\lambda_1/2$.
To prove the other direction, recall that $\B_{(i)}\neq \0$ implies
$\s_i=\B_{(i)}^T/||\B_{(i)}||_2$. Substituting this expression into
\eqref{eq:subgradient2}, rearranging terms, and taking the norm on both sides, we get
   $2||\A^T\X^T\X^{(i)}-\bb_i||_2 =
   \left(2||\X^{(i)}||^2_2+2\lambda+\lambda_1/||\B_{(i)}||_2\right)||\B_{(i)}||_2
  > \lambda_1.$
\end{proof}
\noindent
By Claim~\ref{claim1}, $||\A^T\X^T\X^{(i)}-\bb_i||_2>\lambda_1/2$
  implies that  $\B_{(i)}\neq 0$ which further implies
  $\s_i=\B_{(i)}^T/||\B_{(i)}||_2$. Substituting into \eqref{eq:subgradient2} gives
Step~\ref{step:else} in Algorithm~\ref{alg:groupLasso}.

\vspace{-3mm}
\section{Conclusion}
\label{sec:conclusion}
\vspace{-2mm}

In this paper, we have elucidated several connections between two 
recently-popular matrix decomposition methods that adopt very different 
perspectives on obtaining interpretable low-rank matrix decompositions.
In doing so, we have suggested two optimization problems, \glreg~ and 
\glpca, that highlight similarities and differences between the two 
methods.
In general, SPCA methods obtain interpretability by modifying 
an existing intractable objective with a convex regularization term that 
encourages sparsity, and then \emph{exactly} optimizing that modified 
objective.
On the other hand, CUR methods operate by using randomness and approximation 
as computational resources to optimize \emph{approximately} an intractable
objective, thereby implicitly incorporating a form of regularization into 
the steps of the approximation algorithm.
Understanding this concept of \emph{implicit regularization via approximate 
computation} is clearly of interest more generally, in particular for 
applications where the size scale of the data is expected to increase.

\subsubsection*{Acknowledgments}
We would like to thank Art Owen and Robert Tibshirani for encouragement
and helpful suggestions.  Jacob Bien was supported by the Urbanek
Family Stanford Graduate Fellowship, and Ya Xu was supported by the Melvin 
and Joan Lane Stanford Graduate Fellowship.
In addition, support from the NSF and AFOSR is gratefully acknowledged.

\newpage
% \bibliographystyle{plain}
% \bibliography{mwmbib_jrnl,mwmbib_proc,mwmbib_misc,mwmbib_book,communities}%mwmbib_book

\end{document}